\documentclass[a4paper]{article}

\usepackage{hyperref}
\usepackage[utf8]{inputenc}
\usepackage{algorithm}
\usepackage{algorithmic}
\usepackage{multicol}

\usepackage[margin=1in]{geometry}
\usepackage{graphicx}
\usepackage{float}
\usepackage{subfig}
\usepackage{amsthm,color}
\usepackage{amsmath,amssymb}
\usepackage[dvipsnames]{xcolor}
\usepackage[shortlabels]{enumitem}

\newcommand{\rmv}[1]{}

\newcommand{\MP}[1]{}

\newcommand{\RMP}[1]{}

\newcommand{\str}{\textrm{strong}} 
\newcommand{\sly}{\textrm{strongly}}

\newcounter{NewCounter}
\newcounter{claimcount}[NewCounter]
\newcounter{cclaimcount}[claimcount]

\newenvironment{claim}{
\par\stepcounter{claimcount}\textbf{Claim \arabic{NewCounter}.\arabic{claimcount}}\begin{itshape}
}
{\end{itshape}}

\newenvironment{cclaim}{
\par\stepcounter{cclaimcount}\textbf{Claim \arabic{NewCounter}.\arabic{claimcount}.\arabic{cclaimcount}}\begin{itshape}
}
{\end{itshape}}

\newtheorem{theorem}[NewCounter]{Theorem}

\newtheorem{lemma}[NewCounter]{Lemma}

\newtheorem{definition}[NewCounter]{Definition}

\title{On Linearizability and the\\ Termination of Randomized Algorithms}

\author{Vassos Hadzilacos \qquad Xing Hu \qquad Sam Toueg\\~\\Department of Computer Science\\University of Toronto\\Canada}

\begin{document}
\maketitle

\begin{abstract}
We study the question of whether the ``termination with probability 1'' property
	of a randomized algorithm is preserved when one replaces the atomic registers that the algorithm uses 
	with linearizable (implementations of) registers.
We show that in general this is not so:
	roughly speaking,
	every randomized algorithm $\mathcal{A}$ has a corresponding algorithm $\mathcal{A}'$
	that solves the same problem
	if the registers that it uses are atomic or strongly-linearizable,
	but does not terminate if these registers are replaced with ``merely'' linearizable ones.
Together with a previous result shown in~\cite{abdnotsl},
	this implies that one cannot use the well-known ABD implementation of registers in message-passing systems
	to automatically
	transform any randomized algorithm that works in shared-memory systems
	into a randomized algorithm that works in message-passing systems:
	with a strong adversary the resulting algorithm may not terminate.
\end{abstract}

\section{Introduction}
A well-known property of shared object implementations is \emph{linearizability}~\cite{linearizability}.
Intuitively, with a linearizable object (implementation) each operation must appear as if it takes effect instantaneously
	at some point during the time interval that it actually spans.
As pointed out by the pioneering work of Golab \emph{et al.} \cite{sl11}, however,
	linearizable objects are not as strong as atomic objects in the following sense:
	a randomized algorithm that works with atomic objects may lose some of its properties
	if we replace the atomic objects that it uses with objects that are only linearizable.
In particular, they present a shared-memory randomized algorithm
	that guarantees that
	some random variable has \emph{expected value} 1,
	but if we replace the algorithm's atomic registers with linearizable registers,
	a \mbox{\emph{strong adversary}} can manipulate the schedule to
	ensure that this random variable
	has expected value $\frac{1}{2}$.
To avoid this \mbox{weakness} of linearizability, and ``limit the additional power that a strong adversary may gain when
	atomic objects are replaced with implemented objects'',
	Golab \emph{et al.} introduced the concept of \emph{strong linearizability}~\cite{sl11}.

A natural question is whether this additional power of a strong adversary also applies to
	\emph{termination properties}, more precisely:
	is there a randomized algorithm that 
	(a) terminates with probability 1 against a strong adversary when the objects that it uses are atomic, but
	(b) when these objects are replaced with
	linearizable objects (of the same type),
 	a strong adversary can 
	ensure that
	the algorithm never terminates?
To the best of our knowledge, the question whether the ``termination with probability~1'' property
	can be lost 
	when atomic objects are replaced with linearizable ones
	is not answered by the results in~\cite{sl11},
	or in subsequent papers on this subject~\cite{sl19,sl15,sl12}.
	
This question is particularly interesting because
	one of the main uses of randomized algorithms
	in distributed computing is to achieve termination with probability 1~\cite{abrahamson1988,aspnes1993,aspnes1998,aspnes2003,aspnes1990,aspnes1992,attiya08,bracha1991,chandra1996}
	(e.g., to ``circumvent'' the famous FLP impossibility result~\cite{flp}).
For example, consider the well-known ABD algorithm that implements linearizable shared registers in message-passing systems~\cite{abd}.\footnote{This implementation works under the assumption that fewer than half of the processes may crash.}
One important use of this algorithm is to relate message-passing and shared-memory systems as follows:
	any algorithm that works with atomic shared registers can automatically
	be transformed into an algorithm for message-passing systems
	by replacing its atomic registers with the ABD register implementation.
But can we use the ABD algorithm to automatically transform any shared-memory \emph{randomized} algorithm
	that terminates with probability 1 (e.g., a randomized algorithm that solves consensus)
	into an algorithm that works in message-passing systems?

In this paper, we show that replacing atomic registers with linearizable registers can indeed affect
	the termination property of randomized algorithms: termination with probability 1 can be lost.
In fact we prove that this loss of termination
	is general in the following sense:
	every randomized algorithm~$\mathcal{A}$ has a corresponding algorithm $\mathcal{A}'$
	that solves the same problem
	if the registers that it uses are atomic or strongly-linearizable,
	but does not terminate if these registers are replaced with ``merely'' linearizable ones.
More precisely, we show that for every randomized algorithm $\mathcal{A}$ that solves a task $T$ (e.g., consensus)
	and terminates with probability 1 against a strong adversary,
	there is a corresponding randomized algorithm $\mathcal{A}'$
	that also solves $T$ such that:
	(1)~$\mathcal{A}'$~uses only a set of shared registers in addition to the set of base objects of $\mathcal{A}$;
	(2) if these registers are atomic or $\sly$ linearizable,
             then $\mathcal{A}'$ terminates with probability 1 against a strong adversary,
             and its expected running time is only a small constant more than the expected running time of $\mathcal{A}$;
             but
        (3) if the registers are \emph{only} linearizable, then a strong adversary can prevent the termination of $\mathcal{A}'$.

It is worth noting that this result allows us to answer our previous question about the ABD register implementation,
	namely, whether we can use it to automatically
	transform any randomized algorithm that works in shared-memory systems
	into a randomized algorithm that works in message-passing systems.
In another paper, we proved that, although the registers implemented by the ABD algorithm are linearizable,
	they are \emph{not} strongly linearizable~\cite{abdnotsl}.
Combining this result with the result of this paper proves that,
	in general, using the ABD register implementation instead of atomic registers
	in a randomized algorithm may result in an algorithm that does not terminate.

\section{Model sketch}
We consider a standard asynchronous shared-memory system with \emph{atomic} registers~\cite{lam86,herlihy91}
	where processes are subject to crash failures.
We consider register implementations that are \emph{linearizable}~\cite{linearizability} or \emph{strongly linearizable}~\cite{sl11}.
For brevity, in this paper a  ``linearizable [strongly-linearizable] register'' refers to an ``implemented register whose implementation is linearizable [strongly-linearizable]".

\noindent
The precise definition of strong linearizability of~\cite{sl11} is reproduced here for convenience: 	

\begin{definition}\label{SL}
 A set of histories $\mathcal{H}$ over a set of shared objects is strongly linearizable 
 	if there exists a function $f$ mapping histories 
 	in close($\mathcal{H}$) to sequential histories, such that:

~~~ (L) for any $H \in close(\mathcal{H})$, f(H) is a linearization of H,
	 and 

~~~ (P) for any $G, H \in close(\mathcal{H})$, if G is a prefix of H,
	then f(G) is a prefix of f(H).

\noindent
The function f is called a strong linearization function for $\mathcal{H}$.
\end{definition}

\section{Result}
\begin{algorithm}[!ht]
\caption{Weakener algorithm}
\begin{multicols}{2}
\label{toyalgo}
For $j = 0,1,2,...$
\begin{itemize}
\item $R_1[j]$: MWMR register
	initialized to $\bot$
	
\item $C_1[j]$: SWMR register	
	initialized to $-1$

\item $R_2[j]$: SWMR register initialized to $\textsc{false}$
\end{itemize}
\vspace{0.5cm}
\begin{algorithmic}[1]

\STATE \textsc{Code of process $p_i$, $i\in \{0, 1\}$}:
\FOR  {rounds $j = 0,1,2,...$}
\STATE \{* \textbf{Phase 1:} writing $R_1[j]$ *\}
\STATE $R_1[j] \gets i$ \label{pwrite1}
\IF{$i=0$}\label{cointosser1}
\STATE \{* code executed only by $p_0$ *\}
\STATE $C_1[j] \gets$ flip coin \label{pcoin1}
\ENDIF
\STATE \{* \textbf{Phase 2:} reading $R_2[j]$ *\}
\STATE $v_1 \gets R_2[j]$ \label{v1}
\IF{$v_1 = \textsc{false}$}\label{guard2}
     \STATE \textbf{exit for loop} \label{exit2}
\ENDIF
\ENDFOR
\STATE \textbf{return}\label{halt0}

\columnbreak
~

~

~

~

~

~

~

\STATE \textsc{Code of process $p_i$, \mbox{$i\in \{2, 3, \ldots, n-1\}$}:}
\FOR  {rounds $j = 0,1,2,...$}
\STATE \{* \textbf{Phase 1:} reading $R_1[j]$ and $C_1[j]$ *\}
\STATE $u_1 \gets R_1[j]$ \label{u1}
\STATE $u_2 \gets R_1[j]$ \label{u2}
\STATE $c_1 \gets C_1[j]$ \label{rcoin1} 
\IF{($u_1 \neq c_1$ \OR $u_2 \neq 1-c_1$)}\label{guard1}
     \STATE \textbf{exit for loop} \label{exit1}
\ENDIF
\STATE \{* \textbf{Phase 2:} writing $R_2[j]$ *\}
\STATE $R_2[j] \gets \textsc{true}$ \label{pwrite2}

\ENDFOR
\STATE \textbf{return}\label{halt1}
\end{algorithmic}
\end{multicols}
\end{algorithm}

Consider Algorithm~\ref{toyalgo} for $n \ge 3$ processes $p_0, p_1, p_2, \ldots, p_{n-1}$.
This algorithm uses linearizable registers $R_1[j]$, $R_2[j]$, and $C_1[j]$ for $j \ge 0$.
We first show that if these registers 
	are \emph{not} $\sly$ linearizable,
 	then a strong adversary $\mathcal{S}$ can construct an execution
	of Algorithm~\ref{toyalgo} in which all the processes are \emph{correct}\footnote{A process is correct
	if it takes infinitely many steps. We assume that after returning from the algorithm
	in line~\ref{halt0} or~\ref{halt1}, processes are supposed to take NOP steps (forever).}
	but they loop forever without reaching a return statement in line~\ref{halt0} or~\ref{halt1} (Theorem~\ref{LinearizableIsWeak}).
We then show that if these registers 
	are $\sly$ linearizable, then all the correct processes
	return from the algorithm with probability $1$
	(within 2 rounds in expectation) (Theorem~\ref{WSLinearizableIsStrong}).\footnote{It turns out that
	these results depend only on whether the registers $R_1[j]$ ($j \ge 0$)
	are strongly-linearizable or not; this can be easily seen from
	the proofs of Theorems~\ref{LinearizableIsWeak} and~\ref{WSLinearizableIsStrong}.}

\begin{theorem}\label{LinearizableIsWeak}
If the registers of Algorithm~\ref{toyalgo} are linearizable but not $\sly$ linearizable,
	a~strong adversary $\mathcal{S}$ can construct a run where all the processes execute
	infinitely many rounds (and therefore never return in line~\ref{halt0} or~\ref{halt1}).
\end{theorem}

\begin{proof}
Assume that $R_1[j]$, $R_2[j]$, and $C_1[j]$ (for all $j \ge 0$) are linearizable but \emph{not} $\sly$ linearizable.
A strong adversary $\mathcal{S}$ can construct an infinite execution of Algorithm~\ref{toyalgo} as follows (Figure~\ref{toy}):

\begin{figure}[!htb]
   		 \centering 
    		\includegraphics[width=1\textwidth]{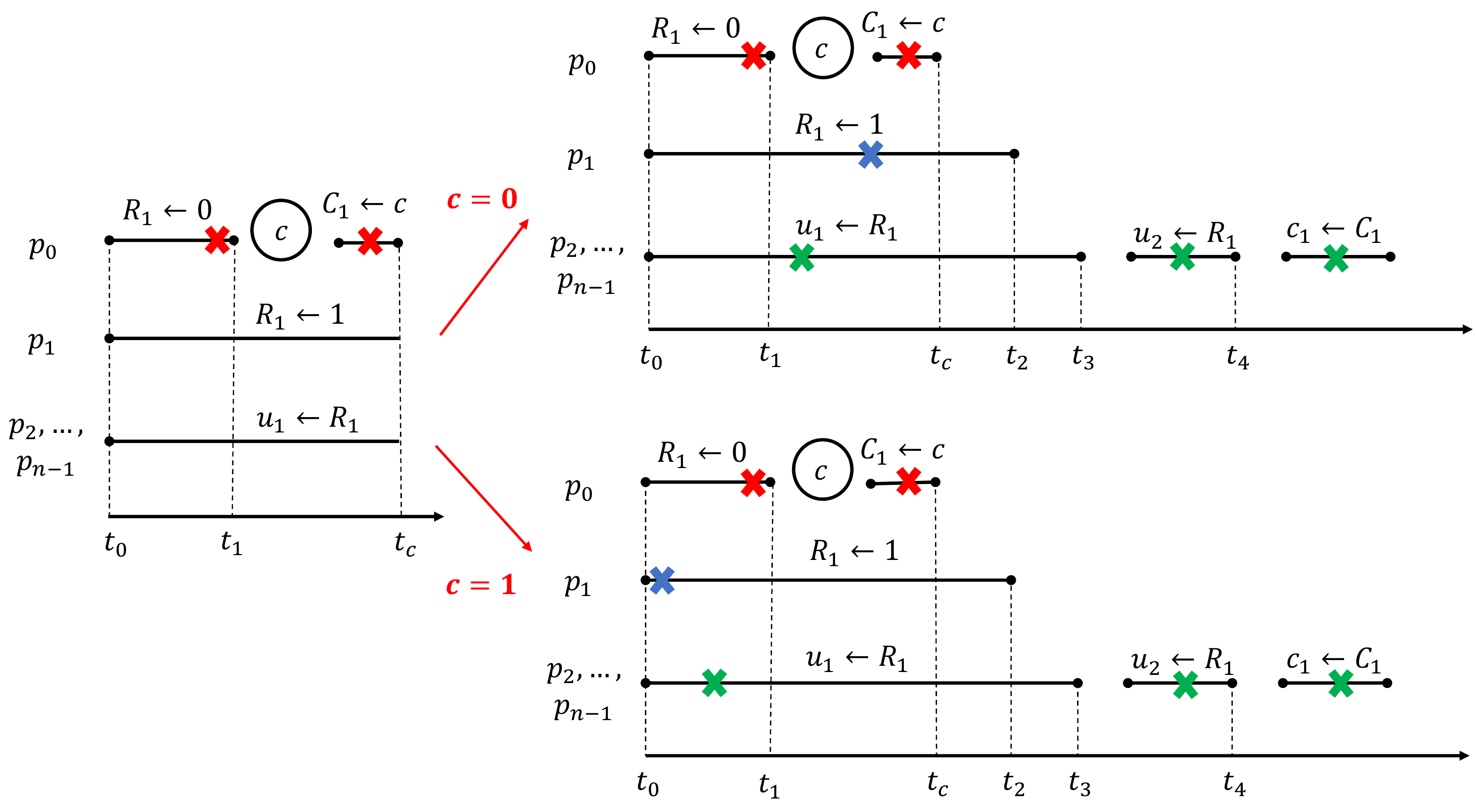}
   		 \caption{Phase 1 in a single round of an infinite execution} 
   		 \label{toy}
	\end{figure}

\newcommand{\procs}{p_2, p_3, \ldots, p_{n-1}}
\newcommand{\Procs}{\textrm{processes } p_2, p_3, \ldots, p_{n-1}}

\begin{enumerate}

\item\label{first-step} \textbf{Phase 1:} At time $t_0$,
	process $p_0$ starts writing~$0$ into $R_1[0]$ in line~\ref{pwrite1},
	process $p_1$ starts writing~$1$ into $R_1[0]$ in line~\ref{pwrite1},
	and $\Procs$ start reading $R_1[0]$ in line~\ref{u1}.

\item At time $t_1 > t_0$, process $p_0$ completes its writing of 0 into $R_1[0]$ in line~\ref{pwrite1}.

\item After time $t_1$,
	process $p_0$ flips a coin and writes the result into $C_1[0]$ in line~\ref{pcoin1}.
	Let $t_c > t_1$ be the time when $p_0$ completes this write.

Depending on the result of $p_0$'s coin flip (and therefore the content of $C_1[0]$),
	the adversary $\mathcal{S}$ continues the run it is constructing in one of the following two ways:

\textbf{Case 1}: $C_1[0]=0$ at time $t_c$.

The continuation of the run in this case is shown at the top of Figure~\ref{toy}.

\begin{enumerate}

\item At time $t_2 > t_c$,
		$p_1$ completes its writing of 1 into $R_1[0]$ (line~\ref{pwrite1}).
		
		Note that \emph{both} $p_0$ and $p_1$ have now completed Phase 1 of round $j=0$.

\item The adversary $\mathcal{S}$ linearizes the write of 1 into $R_1[0]$ by $p_1$
	\emph{after} the write of 0 into $R_1[0]$ by~$p_0$.

\item Note that $\procs$ are still reading $R_1[0]$ in line~\ref{u1}.
Now the adversary linearizes these read operations \emph{between} the above write of 0 by $p_0$ and the write of 1 by~$p_1$.

\item At time $t_3 > t_2$,
		$\Procs$ complete their read of $R_1[0]$ in line~\ref{u1}.
By the above linearization, they read $0$, and so they set (their local variable) $u_1 = 0$ in that line.
		
\item Then $\Procs$ start and complete their read of $R_1[0]$ in line~\ref{u2}.
Since (1)~these reads start \emph{after} the time $t_2$ when $p_1$ completed its write of 1 into $R_1[0]$,
	and (2)~this write is linearized \emph{after} the write of $p_0$ into $R_1[0]$,
	$\Procs$ read~$1$.
So they all set (their local variable) $u_2 = 1$ in line~\ref{u2}.
Let $t_4 > t_3$ be the time when every process $\procs$ has set $u_2 = 1$.

\item After time $t_4$, $\Procs$ start reading $C_1[0]$ in line~\ref{rcoin1}.
		Since $C_1[0] =0$ at time~$t_c$ and it is not modified thereafter,
		$\procs$ read 0 and set (their local variable) $c_1 = 0$ in line~\ref{rcoin1}.
	
\item Then $\procs$ execute line~\ref{guard1} and find that the condition of this line is \emph{not} satisfied
	because they have $u_1 = c_1=0$ and $u_2 = 1- c_1 = 1$.

	   So $\procs$ complete Phase 1 of round $j = 0$
	   \emph{without exiting in line~\ref{exit1}}.
	Recall that both $p_0$ and $p_1$ also completed Phase 1 of round $j=0$ without exiting.

\end{enumerate}

\textbf{Case 2}: $C_1[0]=1$ at time $t_c$.

The continuation of the run in this case is shown at the top of Figure~\ref{toy}.
This continuation is essentially symmetric to the one for Case 1:
	the key difference is that the adversary  $\mathcal{S}$ now linearizes
	the write of $p_1$ before the write of $p_0$, as we describe in detail below.
	
\begin{enumerate}

\item At time $t_2 > t_c$,
		$p_1$ completes its writing of 1 into $R_1[0]$ (line~\ref{pwrite1}).
		
		Note that \emph{both} $p_0$ and $p_1$ have now completed Phase 1 of round $j=0$.

\item $\mathcal{S}$ linearizes the write of 1 into $R_1[0]$ by $p_1$
	\emph{before} the write of 0 into $R_1[0]$ by~$p_0$.

\item Note that $\procs$ are still reading $R_1[0]$ in line~\ref{u1}.
Now the adversary linearizes these read operations \emph{between} the above write of 1 by $p_1$ and the write of 0 by~$p_0$.

\item At time $t_3 > t_2$,
		$\Procs$ complete their read of $R_1[0]$ in line~\ref{u1}.
By the above linearization, they read $1$, and so they set (their local variable) $u_1 = 1$ in that line.
		
\item Then $\Procs$ start and complete their read of $R_1[0]$ in line~\ref{u2}.
Since (1)~these reads start \emph{after} the time $t_1$ when $p_0$ completed its write of 0 into $R_1[0]$,
	and (2)~this write is linearized \emph{after} the write of $p_1$ into $R_1[0]$,
	$\Procs$ read~$0$.
So they all set (their local variable) $u_2 = 0$ in line~\ref{u2}.
Let $t_4 > t_3$ be the time when every process $\procs$ has set $u_2 = 0$.

\item After time $t_4$, $\Procs$ start reading $C_1[0]$ in line~\ref{rcoin1}.
		Since $C_1[0] = 1$ at time~$t_c$ and it is not modified thereafter,
		$\procs$ read 1 and set (their local variable) $c_1 = 1$ in line~\ref{rcoin1}.
	
\item Then $\procs$ execute line~\ref{guard1} and find that the condition of this line is \emph{not} satisfied
	because they have $u_1 = c_1=1$ and $u_2 = 1- c_1 = 0$.

	   So $\procs$ complete Phase 1 of round $j = 0$
	   \emph{without exiting in line~\ref{exit1}}.
	Recall that both $p_0$ and $p_1$ also completed Phase 1 of round $j=0$ without exiting.

\end{enumerate}

Thus in both cases, all $n$ processes complete Phase 1 of round $j = 0$ without exiting,
	and are now poised to execute Phase 2 of this round.
The adversary $\mathcal{S}$ extends the run that it built so far as follows.

\renewcommand{\procs}{p_0 \textrm{ and }  p_1}
\renewcommand{\Procs}{\textrm{processes } p_0 \textrm{ and } p_1}

\item\label{first-step} \textbf{Phase 2:} 
	Process $p_2$ writes~\textsc{true} into $R_2[0]$ in line~\ref{pwrite2};
	let $t'_0$ be the time when this write operation completes.
	(Note that $p_2$ has now completed Phase 2 of round 0.)
	
\item After time $t'_0$, $\Procs$ read $R_2[0]$ into $v_1$ in line~\ref{v1};
	Since $p_2$ completes its write of \textsc{true} into $R_2[0]$
	before $\procs$ start to read this register,
	$\procs$ set $v_1 =  \textsc{true}$ in line~\ref{v1}.

\item Then $\procs$ execute line~\ref{guard2}
	and find that the condition ``$v_1 = \textsc{false}$ ''of this line is \emph{not} satisfied.
         So $\procs$ complete Phase 2 of round $j = 0$
	   \emph{without exiting in line~\ref{exit2}}.

\item Processes ${p_3, \ldots, p_n}$ execute line~\ref{pwrite2},
	and so they also complete Phase 2 of round $j = 0$.
	
So all the $n$ processes $p_0, p_1, \ldots, p_{n-1}$,
	have completed Phase 2 of round $0$ without exiting;
	they are now poised to execute round $j=1$.
\end{enumerate}

\noindent
The adversary $\mathcal{S}$ continues to build the run
	by repeating the above scheduling of $p_0, p_1, \ldots, p_{n-1}$ for rounds $j=1,2,\ldots$.
This gives a non-terminating run of Algorithm~\ref{toyalgo} with \mbox{probability~1:}
	in this run,
	all processes are correct, i.e., each takes an infinite number of steps, but loops forever in a for loop
	and never reaches the return statement that follows this loop (in line~\ref{halt0} or~\ref{halt1}).
\end{proof}

We now prove that if the registers $R_1[j]$ for $j=1,2,...$ are $\sly$ linearizable,
	then Algorithm~\ref{toyalgo} terminates with probability 1, even against a strong adversary.
Roughly speaking,
	this is because if $R_1[j]$ is $\sly$ linearizable,
	then the order in which $0$ and $1$ are written into $R_1[j]$ in line~\ref{pwrite1} is already \emph{fixed}
	before the adversary $\mathcal{S}$ can see result of the coin flip in line~\ref{pcoin1} of round $j$.
So for every round $j\ge 0$,
	the adversary can\emph{not} ``retroactively'' decide on this linearization order
	according to the coin flip result
	(as it does in the proof of Theorem~\ref{LinearizableIsWeak}, where $R_1[j]$ is merely linearizable)
	to ensure that
	processes $p_i$ ($i \ge 2$) do not exit by the condition of line~\ref{guard1}.
Thus, with probability 1/2, all these processes will exit in line~\ref{exit1}.
And if they all exit there, then no process will write \textsc{true} in register $R_2[j]$ in line~\ref{pwrite2},
	and so $p_0$ and $p_1$ will also exit
	in line~\ref{exit2} of round $j$.

\begin{theorem}\label{WSLinearizableIsStrong}
If the registers of Algorithm~\ref{toyalgo} are $\sly$ linearizable,
	 then the algorithm terminates, even against a strong adversary:
	 with probability 1, all the correct processes
	 reach the return statement in line~\ref{halt0} or~\ref{halt1};
	 furthermore, they do so within 2 expected rounds.
\end{theorem}

\noindent
To prove the above theorem, we first show the following two lemmas.
 
\newcommand{\ex}{\textrm{reach}}
\newcommand{\exs}{\textrm{reaches}}

\begin{lemma}\label{smallarestuck}
For all rounds $j \ge0$,
	if no process $\exs$ line~\ref{pwrite2} in round $j$,
	then neither $p_0$ nor $p_1$ enters round $j+1$.
\end{lemma}

\begin{proof}
Suppose no process $\exs$ line~\ref{pwrite2} in round $j$.
Then no process writes into $R_2[j]$, and so $R_2[j] = \textsc{false}$ (the initial value of $R_2[j]$) at all times.
Assume, for contradiction,
	that some process $p_i$ with $i\in\{0,1\}$ enters round $j+1$.
So $p_i$ did not exit in line~\ref{exit2} of round $j$.
Thus, when $p_i$ evaluated the exit condition ``$v_1 = \textsc{false}$'' in line~\ref{guard2} of round $j$,
	it found that $v_1= \textsc{true}$.
But $v_1$ is the value that $p_i$ read from $R_2[j]$ in line~\ref{v1} of that round,
	and so $v_1$ can only be $\textsc{false}$ --- a contradiction.
\end{proof}

\begin{lemma}\label{exitchance}
For all rounds $j \ge0$, 
	with probability at least $1/2$,
	no process enters round $j+1$.
\end{lemma}

\begin{proof}
Consider any round $j\ge0$.
There are two cases:

\begin{enumerate}[(I)]
\item\label{noRwrite} \textbf{Process $p_0$ does not complete its write of register $R_1[j]$ in line~\ref{pwrite1} in round $j$.}

Thus, $p_0$ never reaches line~\ref{pcoin1} (where it writes $C_1[j]$) in round $j$.
So $C_1[j] = -1$ (the initial value of $C_1[j]$) at all times.

\begin{claim}\label{allarestuck-det}
No process enters round $j+1$.
\end{claim}
\begin{proof}
We first show that no process $\exs$ line~\ref{pwrite2} in round $j$,
To see why, suppose, for contradiction,
	some process $p_i$
	$\exs$ line~\ref{pwrite2} in round $j$.
So $p_i$ did not exit in line~\ref{exit1} of round $j$.
Thus, when $p_i$ evaluated the exit condition ($u_1 \neq c_1$ or $u_2 \neq 1-c_1$) in line~\ref{guard1}
	it found the condition to be false, i.e., it found that $u_1= c_1$ and $u_2 = 1-c_1$.
Note that $c_1$ is the value that $p_i$ read from $C_1[j]$ in line~\ref{rcoin1}, and so $c_1 = -1$.
Thus, $p_i$ found that $u_1= -1$ and $u_2 = 2$ in line~\ref{guard1}.
But $u_1$ is the value that $p_i$ read from $R_1[j]$ in line~\ref{u1},
	and so $u_1$ can only be $\bot$ (the initial value of $R_1[j]$),
	or 0~or~1 (the values written into it by $p_0$ and $p_1$, respectively).
So $u_1\neq -1$ in line~\ref{guard1} --- a contradiction.

So no process $\exs$ line~\ref{pwrite2} in round $j$.
This implies that: (i) processes $p_2, \dots, p_{n-1}$ do not enter round $j+1$,
	and (ii) by Lemma~\ref{smallarestuck}, neither $p_0$ nor $p_1$ enters round $j+1$.
\end{proof}

\item\label{Rwrite} \textbf{Process $p_0$ completes its write of register $R_1[j]$ in line~\ref{pwrite1} in round $j$.}

\begin{claim}\label{allarestuck-rndm}
With probability at least 1/2, no process enters round $j+1$.
\end{claim}

\begin{proof}
Consider the set of histories $\mathcal{H}$ of Algorithm~\ref{toyalgo}; this is a set of histories over
	the registers $R_1[j]$, $R_2[j]$, $C_1[j]$ for $j \ge 0$.
Since these registers are strongly linearizable,
	by Lemma 4.8 of~\cite{sl11}, $\mathcal{H}$~is~strongly linearizable, i.e., it has
	at least one strong linearization function
	that satisfies properties (L) and (P) of Definition~\ref{SL}.
Let $f$ be the $\str$ linearization function that
	the~adversary~$\mathcal{S}$~uses.

Let $G$ be an arbitrary history of the algorithm up to
	and including the completion of the write of $0$ into $R_1[j]$
	by $p_0$ in line~\ref{pwrite1} in round $j$.
Since $p_0$ completes its write of $0$ into $R_1[j]$ in $G$,  
	this write operation appears in the $\str$ linearization $f(G)$.
Now there are two cases:

\begin{itemize}
\item\label{Casino1} \textbf{Case A}:
In $f(G)$, the write of $1$ into $R_1[j]$ by $p_1$ in line~\ref{pwrite1} in round $j$
	occurs \emph{before}
the write of $0$ into $R_1[j]$ by $p_0$ in line~\ref{pwrite1} in round $j$.

Since $f$ is a $\str$ linearization function,
	for every extension $H$ of the history $G$
	(i.e., for every history $H$ such that $G$ is a prefix of $H$),
	the write of $1$ into $R_1[j]$
	occurs before
	the write of $0$ into $R_1[j]$
	in the linearization $f(H)$.
Thus, in $G$ and every extension $H$ of $G$,
	no process can first read $0$ from $R_1[j]$ and then read $1$ from~$R_1[j]$~($\star$).

Let $\mathcal{P}$ be the set of processes in $\{p_2 , p_3, \ldots , p_{n-1} \}$
	that reach line~\ref{guard1} in round $j$
	and evaluate the condition
	($u_1 \neq c_1$ or $u_2 \neq 1-c_1$)
	of that line.
Note that $u_1$ and $u_2$ are the values that the processes in $\mathcal{P}$ read from $R_1[j]$
	consecutively in lines~\ref{u1} and \ref{u2}.
So $u_1$ and $u_2$ are in $\{0,1,\bot\}$,
	and, 	by~($\star$), no process can have both $u_1=0$ and $u_2=1$ ($\star \star$).
Moreover, $c_1$ is the value that the processes in $\mathcal{P}$ read from $C_1[j]$ in line~\ref{rcoin1},
	and so $c_1$ is in $\{0,1,-1\}$.

Let $\mathcal{P' \subseteq P}$ be the subset of processes in $\mathcal{P}$
	that have $c_1 = -1$ or $c_1 = 0$ when 
	they evaluate the condition
	($u_1 \neq c_1$ or $u_2 \neq 1-c_1$)
	in line~\ref{guard1} in round $j$.
\begin{cclaim}\label{Mannaggia1}
~
\begin{enumerate}[(a)]
\item\label{SC1} No process in $\mathcal{P'}$ $\exs$ line~\ref{pwrite2} in round $j$.
\item\label{SC2} If $\mathcal{P'} = \mathcal{P}$ then neither $p_0$ nor $p_1$ enters round $j+1$.
\end{enumerate}	
\end{cclaim}

\begin{proof}
To see why (a) holds, note that:
	(i) every process $p_i$ in $\mathcal{P'}$ that has $c_1 = -1$
	evaluates the condition ($u_1 \neq c_1$ or $u_2 \neq 1-c_1$) in line~\ref{guard1} to true because $u_1 \neq -1$;
	and
	(ii) every process $p_i$ in $\mathcal{P'}$ that has $c_1 = 0$, 
	also evaluates the condition ($u_1 \neq c_1$ or $u_2 \neq 1-c_1$) in line~\ref{guard1} to true
	(otherwise $p_i$ would have both $u_1 = c_1 = 0$ and $u_2 = 1-c_1 = 1$, which is not possible by ($\star \star$)).
Thus, no process $p_i$ in $\mathcal{P'}$ $\exs$ line~\ref{pwrite2} in round $j$
	(it would exit in line~\ref{exit1} before reaching that line).

To see why (b) holds, suppose $\mathcal{P'} = \mathcal{P}$
	and consider an arbitrary process $p$.
If $p \not \in \mathcal{P}$
	then $p$ does not evaluate the condition
	($u_1 \neq c_1$ or $u_2 \neq 1-c_1$)
	in line~\ref{guard1} in round $j$;
	and if $p \in \mathcal{P}$, then $p \in \mathcal{P'}$, and so from part (a),
	$p$ does not $\ex$ line~\ref{pwrite2} in round $j$.
So in both cases,
	$p$ does not $\ex$ line~\ref{pwrite2} in round $j$.
Thus no process $\exs$ line~\ref{pwrite2} in round $j$,
	and so, by Lemma~\ref{smallarestuck}, neither $p_0$ nor $p_1$ enters round $j+1$.
\end{proof}

Now recall that $G$ is the history of the algorithm up to
	and including the completion of the write of $0$ into $R_1[j]$
	by $p_0$ in line~\ref{pwrite1} in round $j$.
After this write, i.e., in any extension $H$ of $G$,
	$p_0$ is supposed to flip a coin and write the result into $C_1[j]$ in line~\ref{pcoin1}.
Thus, with probability \emph{at least}~$1/2$,
	$p_0$~will \emph{not} invoke
	the operation to write $1$ into $C_1[j]$.
So with probability at least~$1/2$,
	processes never read 1 from $C_1[j]$.
Thus with probability at least~$1/2$, no process in $\mathcal{P}$ has $c_1 = 1$
	when it evaluates the condition
	($u_1 \neq c_1$ or $u_2 \neq 1-c_1$)
	in line~\ref{guard1} in round~$j$.
Since $c_1 \in \{0,1,-1\}$, this implies that
	with probability at least~$1/2$,
	every process in $\mathcal{P}$ has $c_1 = -1$ or $c_1 = 0$
	when it evaluates this condition
	in line~\ref{guard1} in round $j$;
	in other words,
        with probability at least~$1/2$,
	$\mathcal{P' = P}$.
Therefore, from Claim~\ref{Mannaggia1}, with probability at least~$1/2$:
\begin{enumerate}[(a)]
\item\label{SC1} No process in $\mathcal{P}$ $\exs$ line~\ref{pwrite2} in round $j$.
\item\label{SC2} Neither $p_0$ nor $p_1$ enters round $j+1$.
\end{enumerate}
This implies that in Case~A, with probability (at least)~$1/2$,
	no process enters round $j+1$.

\item\label{Casino2} \textbf{Case B}:
In $f(G)$, the write of $1$ into $R_1[j]$ by $p_1$ in line~\ref{pwrite1} in round $j$
	does \emph{not} occur before
the write of $0$ into $R_1[j]$ by $p_0$ in line~\ref{pwrite1} in round $j$.
This case is essentially symmetric to the one for Case~A, we include it below for completeness.

Since $f$ is a $\str$ linearization function,
	for every extension $H$ of the history $G$,
	the write of~$1$ into $R_1[j]$
	does not occur before
	the write of $0$ into $R_1[j]$
	in the linearization $f(H)$.
Thus, in $G$ and every extension $H$ of $G$,
	no process can first read $1$ from $R_1[j]$ and then read $0$ from~$R_1[j]$~($\dagger$).

Let $\mathcal{P}$ be the set of processes in $\{p_2 , p_3, \ldots , p_{n-1} \}$
	that reach line~\ref{guard1} in round $j$
	and evaluate the condition
	($u_1 \neq c_1$ or $u_2 \neq 1-c_1$)
	of that line.
Note that $u_1$ and $u_2$ are the values that the processes in $\mathcal{P}$ read from $R_1[j]$
	consecutively in lines~\ref{u1} and \ref{u2}.
So $u_1$ and $u_2$ are in $\{0,1,\bot\}$,
	and, 	by ($\dagger$), no process can have both $u_1=1$ and $u_2=0$ $(\dagger \dagger)$.
Moreover, $c_1$ is the value that the processes in $\mathcal{P}$ read from $C_1[j]$ in line~\ref{rcoin1},
	and so $c_1$ is in $\{0,1,-1\}$.

Let $\mathcal{P' \subseteq P}$ be the subset of processes in $\mathcal{P}$
	that have $c_1 = -1$ or $c_1 = 1$ when 
	they evaluate the condition
	($u_1 \neq c_1$ or $u_2 \neq 1-c_1$)
	in line~\ref{guard1} in round $j$.

\begin{cclaim}\label{Mannaggia2}
~
\begin{enumerate}[(a)]
\item\label{SCC1} No process in $\mathcal{P'}$ $\exs$ line~\ref{pwrite2} in round $j$.
\item\label{SCC2} If $\mathcal{P'} = \mathcal{P}$ then neither $p_0$ nor $p_1$ enters round $j+1$.
\end{enumerate}	
\end{cclaim}

\begin{proof}
To see why (a) holds, note that:
	(i) every process $p_i$ in $\mathcal{P'}$ that has $c_1 = -1$
	evaluates the condition ($u_1 \neq c_1$ or $u_2 \neq 1-c_1$) in line~\ref{guard1} to true because $u_1 \neq -1$;
	and
	(ii) every process $p_i$ in $\mathcal{P'}$ that has $c_1 = 1$, 
	also evaluates the condition ($u_1 \neq c_1$ or $u_2 \neq 1-c_1$) in line~\ref{guard1} to true
	(otherwise $p_i$ would have both $u_1 = c_1 = 1$ and $u_2 = 1-c_1 = 0$, which is not possible by ($\dagger \dagger$)).
Thus, no process $p_i$ in $\mathcal{P'}$ $\exs$ line~\ref{pwrite2} in round $j$
	(it would exit in line~\ref{exit1} before reaching that line).
	
To see why (b) holds, suppose $\mathcal{P'} = \mathcal{P}$
	and consider an arbitrary process $p$.
If $p \not \in \mathcal{P}$
	then $p$ does not evaluate the condition
	($u_1 \neq c_1$ or $u_2 \neq 1-c_1$)
	in line~\ref{guard1} in round $j$;
	and if $p \in \mathcal{P}$, then $p \in \mathcal{P'}$, and so from part (a),
	$p$ does not $\ex$ line~\ref{pwrite2} in round $j$.
So in both cases,
	$p$ does not $\ex$ line~\ref{pwrite2} in round $j$.
Thus no process $\exs$ line~\ref{pwrite2} in round $j$,
	and so, by Lemma~\ref{smallarestuck}, neither $p_0$ nor $p_1$ enters round $j+1$.
\end{proof}

Now recall that $G$ is the history of the algorithm up to
	and including the completion of the write of $0$ into $R_1[j]$
	by $p_0$ in line~\ref{pwrite1} in round $j$.
After this write, i.e., in any extension $H$ of $G$,
	$p_0$ is supposed to flip a coin and write the result into $C_1[j]$ in line~\ref{pcoin1}.
Thus, with probability \emph{at least}~$1/2$,
	$p_0$~will \emph{not} invoke
	the operation to write $0$ into $C_1[j]$.
So with probability at least~$1/2$,
	processes never read 0 from $C_1[j]$.
Thus with probability at least~$1/2$, no process in $\mathcal{P}$ has $c_1 = 0$
	when it evaluates the condition
	($u_1 \neq c_1$ or $u_2 \neq 1-c_1$)
	in line~\ref{guard1} in round~$j$.
Since $c_1 \in \{0,1,-1\}$, this implies that
	with probability at least~$1/2$,
	every process in $\mathcal{P}$ has $c_1 = -1$ or $c_1 = 1$
	when it evaluates this condition
	in line~\ref{guard1} in round $j$;
	in other words,
        with probability at least~$1/2$,
	$\mathcal{P' = P}$.
Therefore, from Claim~\ref{Mannaggia2}, with probability at least~$1/2$:

\begin{enumerate}[(a)]
\item\label{SC1} No process in $\mathcal{P}$ $\exs$ line~\ref{pwrite2} in round $j$.
\item\label{SC2} Neither $p_0$ nor $p_1$ enters round $j+1$.
\end{enumerate}
This implies that in Case~B, with probability at least~$1/2$,
	no process enters round $j+1$.
\end{itemize}
So in both Cases A and B, with probability at least 1/2, no process enters round $j+1$.
\end{proof}
\end{enumerate}

\noindent
Therefore, from Claims~\ref{allarestuck-det} and~\ref{allarestuck-rndm} of Cases~\ref{noRwrite} and~\ref{Rwrite},
	with probability at least $1/2$,
	no process enters round $j+1$.
\end{proof}

\noindent
We can now complete the proof of Theorem~\ref{WSLinearizableIsStrong},
	namely, that with $\sly$ linearizable registers,
	Algorithm~\ref{toyalgo} terminates with probability 1 in expected $2$ rounds, even against a strong adversary.
%
\noindent
Consider any round $j \ge 0$.
By Lemma~\ref{exitchance},
	with probability at least $1/2$,
	no process enters round $j+1$.
Since this holds for every round $j \ge 0$,
	then it must be that, with probability 1,
	all the processes that take an infinite number of steps
	must exit their loop in lines~\ref{exit1} or \ref{exit2},
	and reach the return statement that follows this loop;
	furthermore, they do so within~2 expected iterations of the loop.

\begin{theorem}\label{bigbob}
Let $\mathcal{A}$ be any randomized algorithm that solves a task $T$ (such as consensus)
	for \linebreak \mbox{$n \ge 3$} processes and terminates with probability 1 against a strong adversary.
There is a corresponding randomized algorithm $\mathcal{A}'$
	that solves $T$ for $n \ge 3$ processes such that:

\begin{enumerate}

\item $\mathcal{A}'$ uses a set $\cal{R}$ of shared registers in addition to the set of base objects of $\mathcal{A}$.

\item If the registers in $\cal{R}$ are atomic or $\sly$ linearizable,
    then $\mathcal{A}'$ terminates with probability 1 against a strong adversary.  
    Furthermore, the expected running time of $\mathcal{A}'$ is only a constant
    more than the expected running time of $\mathcal{A}$.

\item If the registers in $\cal{R}$ are \emph{only} linearizable,
    then a strong adversary can prevent the termination~of~$\mathcal{A}'$.
    
\end{enumerate}

\end{theorem}

\begin{proof}
Consider any randomized algorithm $\mathcal{A}$ that solves some task $T$
	for $n \ge 3$ processes $p_0, p_1, p_2, \ldots, p_{n-1}$,
	and terminates with probability 1 against a strong adversary.
Using $\mathcal{A}$, we construct the following randomized algorithm $\mathcal{A}'$:
	every process $p_i$ with $i\in\{0,1,2,...,n-1\}$
	first executes Algorithm~\ref{toyalgo};
	if $p_i$ returns then it executes algorithm $\mathcal{A}$.
Note that:

\begin{enumerate}

\item In addition to the set of base objects that $\mathcal{A}$ uses,
	the algorithm $\mathcal{A}'$ uses the set of shared registers ${\cal{R}} = \{R_1[j], R_2[j], C_1[j] ~|~ \textrm{for } j \ge 0 \}$.
  
\item Suppose these registers are $\sly$ linearizable.
	Then, by Theorem~\ref{WSLinearizableIsStrong},
	Algorithm~\ref{toyalgo} (that processes execute before executing $\mathcal{A}$)
	terminates with probability 1 in expected $2$ rounds against a strong adversary.
	Since $\mathcal{A}$ also terminates with probability 1 against a strong adversary,
	the algorithm $\mathcal{A}'$ also terminates with probability 1 against a strong adversary,
	and the expected running time of $\mathcal{A}'$ is only a constant time
        more than the expected running time of the given algorithm $\mathcal{A}$.

	Since $\mathcal{A}$ solves task $T$, it is clear that $\mathcal{A}'$ also solves $T$.
	
\item Suppose these registers are linearizable but not $\sly$ linearizable.
	Then, by Theorem~\ref{LinearizableIsWeak},
	a~strong adversary can construct a run of Algorithm~\ref{toyalgo} where, with probability~$1$,
	all the processes execute
	infinitely many rounds and never return.
	Thus, since $\mathcal{A}'$ starts by executing Algorithm~\ref{toyalgo},
	it is clear that a strong adversary can prevent the termination of $\mathcal{A}'$ \emph{with probability 1}.
	\qedhere
\end{enumerate}
\end{proof}

\bibliographystyle{abbrv}
\bibliography{simpler}

\begin{thebibliography}{10}

\bibitem{abrahamson1988}
K.~Abrahamson.
\newblock On achieving consensus using a shared memory.
\newblock In {\em Proceedings of the seventh annual ACM Symposium on Principles
  of distributed computing}, pages 291--302, 1988.

\bibitem{aspnes1993}
J.~Aspnes.
\newblock Time-and space-efficient randomized consensus.
\newblock {\em Journal of Algorithms}, 14(3):414--431, 1993.

\bibitem{aspnes1998}
J.~Aspnes.
\newblock Lower bounds for distributed coin-flipping and randomized consensus.
\newblock {\em Journal of the ACM (JACM)}, 45(3):415--450, 1998.

\bibitem{aspnes2003}
J.~Aspnes.
\newblock Randomized protocols for asynchronous consensus.
\newblock {\em Distributed Computing}, 16(2-3):165--175, 2003.

\bibitem{aspnes1990}
J.~Aspnes and M.~Herlihy.
\newblock Fast randomized consensus using shared memory.
\newblock {\em Journal of Algorithms}, 11(3):441 -- 461, 1990.

\bibitem{aspnes1992}
J.~Aspnes and O.~Waarts.
\newblock Randomized consensus in expected {$O(n(\log n)^2$)} operations per
  processor.
\newblock In {\em Annual Symposium on Foundations of Computer Science},
  volume~33, pages 137--137. Citeseer, 1992.

\bibitem{abd}
H.~Attiya, A.~Bar-Noy, and D.~Dolev.
\newblock Sharing memory robustly in message-passing systems.
\newblock {\em J.~ACM}, 42(1):124--142, Jan. 1995.

\bibitem{attiya08}
H.~Attiya and K.~Censor.
\newblock Tight bounds for asynchronous randomized consensus.
\newblock {\em J. ACM}, 55(5), Nov. 2008.

\bibitem{sl19}
H.~Attiya and C.~Enea.
\newblock {Putting Strong Linearizability in Context: Preserving
  Hyperproperties in Programs That Use Concurrent Objects}.
\newblock In {\em 33rd International Symposium on Distributed Computing, {DISC}
  2019}, pages 2:1--2:17, 2019.

\bibitem{bracha1991}
G.~Bracha and O.~Rachman.
\newblock Randomized consensus in expected {$O(n^2 \log n)$} operations.
\newblock In {\em International Workshop on Distributed Algorithms}, pages
  143--150. Springer, 1991.

\bibitem{chandra1996}
T.~D. Chandra.
\newblock Polylog randomized wait-free consensus.
\newblock In {\em Proceedings of the fifteenth annual ACM symposium on
  Principles of distributed computing}, pages 166--175, 1996.

\bibitem{sl15}
O.~Denysyuk and P.~Woelfel.
\newblock Wait-freedom is harder than lock-freedom under strong
  linearizability.
\newblock In Y.~Moses, editor, {\em 29rd International Symposium on Distributed
  Computing, {DISC} 2015}, pages 60--74, Berlin, Heidelberg, 2015. Springer
  Berlin Heidelberg.

\bibitem{flp}
M.~J. Fischer, N.~A. Lynch, and M.~S. Paterson.
\newblock Impossibility of distributed consensus with one faulty process.
\newblock {\em J. ACM}, 32(2):374–382, Apr. 1985.

\bibitem{sl11}
W.~Golab, L.~Higham, and P.~Woelfel.
\newblock Linearizable implementations do not suffice for randomized
  distributed computation.
\newblock In {\em Proceedings of the Forty-Third Annual ACM Symposium on Theory
  of Computing}, STOC ’11, page 373–382, New York, NY, USA, 2011.
  Association for Computing Machinery.

\bibitem{abdnotsl}
V.~Hadzilacos, X.~Hu, and S.~Toueg.
\newblock On atomic registers and randomized consensus in {M{\&}M} systems,
  2020.

\bibitem{sl12}
M.~Helmi, L.~Higham, and P.~Woelfel.
\newblock Strongly linearizable implementations: Possibilities and
  impossibilities.
\newblock In {\em Proceedings of the 2012 ACM Symposium on Principles of
  Distributed Computing}, PODC '12, page 385–394, New York, NY, USA, 2012.
  Association for Computing Machinery.

\bibitem{herlihy91}
M.~Herlihy.
\newblock Wait-free synchronization.
\newblock {\em ACM Trans. Program. Lang. Syst.}, 13(1):124–149, Jan. 1991.

\bibitem{linearizability}
M.~P. Herlihy and J.~M. Wing.
\newblock Linearizability: A correctness condition for concurrent objects.
\newblock {\em ACM Trans. Program. Lang. Syst.}, 12(3):463–492, July 1990.

\bibitem{lam86}
L.~Lamport.
\newblock On interprocess communication {P}arts {I}--{II}.
\newblock {\em Distributed Computing}, 1(2):77--101, May 1986.

\end{thebibliography}

\end{document}